\documentclass[11pt]{article}
\usepackage{amsmath,amsthm,amssymb,hyperref,enumerate,xcolor}
\usepackage[ruled]{algorithm2e}
\newcommand{\remove}[1]{}
\newtheorem{thm}{Theorem}[section]  

\newtheorem{lem}[thm]{Lemma}
\newtheorem{define}[thm]{Definition}
\newtheorem{cor}[thm]{Corollary}
\newtheorem{obs}[thm]{Observation}

\newtheorem{remark}[thm]{Remark}

\def\F{{\mathbb{F}}}

\def\N{{\mathbb{N}}}

\def\cS{{\mathcal S}}

\def\cM{\mathcal M}
\def\cH{\mathcal H}

\def\ba{{\mathbf a}}
\def\bb{{\mathbf b}}
\def\bc{{\mathbf c}}
\def\be{{\mathbf e}}

\def\bv{{\mathbf v}}
\def\bx{{\mathbf x}}
\def\by{{\mathbf y}}
\def\prob{{\mathbf{Pr}}}

\def\_{\,\,\,\,\,}

\def\st{\text{ s.t. }}

\def\then{\Rightarrow}

\def\D{{\partial}}

\def\poly{\textsf{poly}}

\def\char{\textsf{char}}

\newcommand{\eps}{\epsilon}
\newcommand{\veps}{\varepsilon}

\newcommand{\ds}{\displaystyle}

\newcommand{\ignore}[1]{}

\allowdisplaybreaks

\begin{document}

\title{Testing Equivalence of Polynomials under Shifts}

\author{Zeev Dvir\thanks{Department of Computer Science and Department of Mathematics,
Princeton University.
Email: \texttt{zeev.dvir@gmail.com}. Research  supported by NSF grants CCF-1217416
and CCF-0832797.} \and
Rafael Oliveira\thanks{Department of Computer Science, Princeton University.
Email: \texttt{rmo@cs.princeton.edu}.} \and
Amir Shpilka\thanks{Department of Computer Science, Technion --- Israel Institute of Technology, Haifa, Israel,
\texttt{shpilka@cs.technion.ac.il}.  The research leading to these results has received funding
from the European Community's Seventh Framework Programme (FP7/2007-2013) under grant agreement number 257575.}}

\date{}
\maketitle

\thispagestyle{empty}
\pagenumbering{arabic}

\begin{abstract}

Two polynomials $f, g \in \F[x_1, \ldots, x_n]$ are called {\em shift-equivalent} if there exists a vector
$(a_1, \ldots, a_n) \in {\F}^n$ such that the polynomial identity $f(x_1+a_1, \ldots, x_n+a_n) \equiv g(x_1,\ldots,x_n)$
holds. Our main result is a new randomized algorithm that tests whether two given polynomials are shift equivalent.
Our algorithm runs in time polynomial in the circuit size of the polynomials, to which it is given black box access. This
complements  a previous work of Grigoriev \cite{grig-q2tcs} who gave a deterministic algorithm running in time $n^{O(d)}$
for degree $d$ polynomials.
 
Our algorithm uses randomness only to solve instances of the Polynomial Identity Testing (PIT) problem. Hence, if one
could de-randomize PIT (a long-standing open problem in complexity) a de-randomization of our algorithm would follow.
This establishes an equivalence between de-randomizing shift-equivalence testing and de-randomizing PIT (both in the black-box and the white-box setting). For certain
restricted models, such as Read Once Branching Programs, we already obtain a deterministic algorithm using existing
PIT results.

\end{abstract}

\thispagestyle{empty}

\newpage
\pagenumbering{arabic}

\section{Introduction}

In this paper we address the following problem, which we call {\em Shift Equivalence Testing} (SET). Given two polynomials
$f,g\in \F[\bx]$ (we use boldface letters to denote vectors), decide whether there exists a shift $\ba\in \F^n$ such that
$f(\bx+\ba)\equiv g(\bx)$ and output one if it exist. The symbol $\equiv$ is used to denote polynomial identity (the
polynomials should have the same coefficients). We will focus mainly on the case where $\F$ is a field of characteristic zero
(such as the rational numbers) or has a sufficiently large positive characteristic.

Observe that $f$ is
shift-equivalent to the zero polynomial if and only if $f$ itself is the zero polynomial. Hence, SET is a natural  generalization of
the well-known Polynomial Identity Testing problem (PIT) in which we need to test whether $f(\bx)\equiv 0$ given access to a
succinct  representation of $f$ (say, as a circuit). A classical randomized algorithm by Schwartz-Zippel-DeMillo-Lipton
\cite{Schwartz80,Zippel79,DemilloL78} is known for PIT:  evaluate $f$ on a random input (from a large enough domain)
and test if $f$ evaluates to zero on that point. If $f$ is non-zero, then it is not zero on a random point with very high probability.
In contrast, it is not clear at all how to devise a
randomized algorithm for SET. Unlike PIT, which is a `co-NP' type problem (there is short proof that a polynomial is {\em not}
zero), the SET problem is an `RP$^\text{NP}$' type problem (there is a short witness (the shift itself) that polynomials 
{\em are} shift equivalent, and verifying that witness is in RP).

The problem of equivalence of polynomials under shifts of the input first appeared in the works of Grigoriev, Lakshman,
Saunders and Karpinski \cite{grig-interp, grig-sparseshift, lak-shiftsuni} (see also references therein),
in the context of finding sparse shifts
of a polynomial.  That is, they were interested in finding a shift that will make a given polynomial  sparse, if such a shift
indeed exists. The main motivation for this question comes from considering polynomials in their sum-of-monomials
representation (also called dense representation or depth-$2$ circuit complexity), and the goal is to find a shift that will
make the representation more succinct. Later,  in \cite{grig-q2tcs}, Grigoriev asked
the following question: given two polynomials $f, g \in \F[\bx]$, is there an
efficient algorithm that can find whether there exists a shift $\ba \in {\F}^n$ such that
$f(\bx+\ba) \equiv g(\bx)$? In the same paper, Grigoriev gave algorithms for three versions
of this problem: one deterministic for characteristic zero, one randomized for large enough characteristic $0<p$ and
one quantum for characteristic $2$. The running time of Grigoriev's algorithms was polynomial in the dense representation.
That is, for polynomials of degree $d$ in $n$ variables, the running time was $n^{O(d)}$ (which is an upper bound on the
number of coefficients). In this paper, we address the same question as Grigoriev, but assume that the polynomials are given
in some succinct representation (say, as arithmetic circuits). In this representation, one can hope for running time which is
polynomial in the size of the given circuits (which can be exponentially small relative to the dense representation). For example
the determinant polynomial has $n^2$ variables and degree $n$ but can be given as a circuit of size $n^{O(1)}$ in the succinct
representation.

Our main result is a new randomized (two-sided error) algorithm for SET. The algorithm runs in time polynomial in the circuit
size of the given polynomials. In fact, we only require {\em black-box} access to the polynomials $f$ and $g$ and a bound on
their degree and circuit size. Our algorithm is obtained as a reduction to the PIT problem. Hence, if we were able to perform
deterministic PIT, we could also perform deterministic SET. For certain interesting restricted models of arithmetic computation,
this already gives deterministic SET. For general circuits, our results show that it is equivalently hard to de-randomize PIT and
SET, which is somewhat surprising as by the explanation above it seems as if SET is a much harder problem than PIT.

Below, we will state our results in the most general way, assuming $f$ and $g$ belong to some circuit classes closed under
certain operations.  The reason for doing this is that, in this way, one can see exactly what conditions are required to
de-randomize the algorithm. That is, what kind of deterministic PIT is required to derive deterministic SET (in general we
require PIT for a slightly larger class). Before giving a formal description of our results we take a moment to set up some
necessary background on PIT and hitting sets.

%%%%%%%%%%%%%%%%%%%%%%%%%%%%%%%%%%%%%%%%%%%%%%%%
%%%%%%%%%%%%%%%%%%%%%%%%%%%%%%%%%%%%%%%%%%%%%%%%

\subsection{PIT and Hitting Sets}

We start by formally defining arithmetic circuits. For more background on arithmetic computation and arithmetic complexity
we refer the reader to the survey \cite{amir-amir}.

\begin{define}[Arithmetic circuit] \label{def:circuit}
An \emph{arithmetic circuit} $C$ is a directed acyclic labeled graph in which the vertices are called `gates'. The gates of
$C$ with in-degree $0$ are called {\em inputs} and are labeled by either a variable from $\{x_1, \ldots, x_n\}$ or by field
element from $\F$. Every other gate of $C$ is labeled by either `$\times$' or `$+$' and has in-degree 2. There is one
gate with out-degree 0, which we call the output gate. Each gate in $C$ computes a polynomial in $\F[\bx]$ in the natural way.
We call the polynomial computed at the output gate `the polynomial computed by $C$'. An arithmetic circuit is called
a \emph{formula} if its underlying graph is a tree.
\end{define}

The PIT problem is defined as follows: we are given an arithmetic circuit $C$ computing a polynomial $f \in \F[\bx]$, and we
have to determine whether the polynomial $f$ is the zero polynomial or not. PIT is a central problem in algebraic complexity.
Deterministically solving PIT is known to imply lower bounds for arithmetic circuits
\cite{HeintzSchnorr80,Agrawal05,KabanetsImpagliazzo04,DSY09}. PIT also has some algorithmic implications. The famous
AKS primality test \cite{AKS04} is based on solving PIT for a specific polynomial. Randomized algorithms for finding a
perfect matching in a given graph reduce the problem to PIT of certain determinants with variables as entries
\cite{Lovasz79,KUW86,MVV87}.

In recent years, there has been considerable progress on the problem of obtaining deterministic PIT algorithms for restricted
classes of circuits. The study of restricted models began with the class of sparse polynomials, which are also referred to as
depth 2 circuits (of the form $\Sigma\Pi$). A long line of work, culminating in the algorithm of Klivans and Spielman
\cite{KlivansSpielman01} gives deterministic PIT for sparse polynomials. In the past decade a series of algorithms
\cite{DvirShpilka06,KayalSaxena07,KarninShpilka08, SaxenaSeshadhri09, KayalSaraf09, 
SaxenaSeshadhri10,AgrawalSSS12} were
devised to solve PIT for circuits of depth 3 with bounded fan-in, which are denoted as $\Sigma\Pi\Sigma(k)$ circuits. A more 
recent line of work, to which we will go back later in the paper, deals with read once branching programs and low rank tensors
\cite{ForbesShpilka12,ForbesShpilka12a,ForbesSS13}.

There are two variants of the PIT problem: in the white-box model the PIT algorithm is given as input an actual arithmetic circuit 
computing $f \in \F[\bx]$ and has to determine if $f \equiv 0$, possibly by inspecting the structure of the circuit. In the
(harder) black-box model, we can only access the polynomial $f$ by querying its value at points $\ba \in \F^n$ of
our choice (we are still assuming $f$ {\em has} some small circuit). It is not hard to see that any deterministic black-box PIT 
algorithm works by evaluating $f$ on some fixed set of
points and outputs $f\equiv 0$ iff all of these evaluations result in zero.  Such a set of evaluation points is called a
\emph{Hitting Set} for the class of circuits to which $f$ is assumed to belong. It is clear that solving PIT in the black-box model
is at least as hard as solving it in the white-box model and indeed, in some cases we have better algorithms in the white-box
model than in the black-box model (compare e.g. \cite{RazShpilka05} to \cite{ForbesShpilka12a} and \cite{ForbesSS13}).

More formally, to deterministically solve  black-box PIT for a class of circuits $\cM$, we need to be able to generate
a hitting set $\cH$ such that for each non-zero polynomial $f$ computed by a circuit in $\cM$, there exists a point
$\ba \in \cH$ such that $f(\ba) \neq 0$. If this is the case, we say that the set $\cH$ \emph{hits}
the class $\cM$, and that the point $\ba$ \emph{hits} $f$.

The following folklore result shows that there {\em exists} a small hitting set for the class of poly-size circuits  (see
Theorem 4.3 of  \cite{amir-amir} for a proof).

\begin{thm}[Non-constructive hitting sets]\label{thm:non constructive hitting set}
For every $n,d,s$ and a field $\F$ of size $|\F| \geq \max(d^2,s)$, there exists a set $\cH \subseteq \F^n$ of size
$|\cH|=\poly(d,s)$ that is a hitting set for all circuits of size at most $s$ and degree at most $d$.
Furthermore, a random set $\cH$ of the appropriate size is such a hitting set with high probability.
\end{thm}

We remark that the theorem above requires that the field size is at least polynomially larger than some of the parameters.
This is necessary for constructing hitting sets since two non-identical polynomials might evaluate to the same value on
all inputs from a sufficiently small sub-field (e.g., $x = x^p$ in $\F_p$). For simplicity, we will assume that we work over
sufficiently large finite fields (so that they contain a hitting set), if necessary by going to an extension field. When working
over characteristic zero  we will implicitly assume that all
constants involved in the hitting sets or in the computation have polynomially long bit representation so we do not have to
keep track of that measure as well. This is quite reasonable given that explicit constructions of hitting sets have this property
and that we can achieve this with randomized constructions as well.

Now, if what we want is to hit only a specific nonzero polynomial, then we do not need the full power of a hitting set. 
As we mentioned before, the randomized algorithm by Schwartz-Zippel-DeMillo-Lipton \cite{Schwartz80,Zippel79,DemilloL78}
gives us a point that hits a given nonzero polynomial with high probability. More formally we have:

\begin{lem}[\cite{Schwartz80,Zippel79,DemilloL78}]\label{lem:sz}
	Let $f(x_1, \ldots, x_n) \in \F[x_1, \ldots, x_n]$ be a nonzero polynomial of degree at most $d$, and let 
	$T \subseteq \F$ be a finite set.
	If we choose $\ba = (a_1, \ldots, a_n) \in T^n$ uniformly at random, then $\prob[f(a) = 0] \le d/|T|$.
\end{lem}

Notice that, to achieve error at most $\veps$ with this lemma, we should pick a set $T$ of size $|T| \ge d/\veps$. 
Generating such a uniformly random element $\ba$ from $T^n$ requires $n \cdot \lceil \log(d/\veps) \rceil$ random bits.

%%%%%%%%%%%%%%%%%%%%%%%%%%%%%%%%%%%%%%%%%%%%%%%%
\subsection{Formal statement of our results}

Our results rely on closure properties of the underlying circuit classes.

\begin{define}\label{def:close}
Given a class of arithmetic circuits $\cM$ we will say that $\cM$ is closed under an operator $A : \F[\bx] \mapsto \F[\bx]$ if the
following property holds. Let $f$ be an $n$-variate polynomial of total degree $d$ that is computed by a circuit of
size $s$ from $\cM$. Then we require that $A(f)$ is computed by a
circuit of size $\poly(n,d,s)$ from $\cM$.
\end{define}

For instance, one operator that is very common and under which all of the most studied circuit classes are closed is 
the restriction operator, namely, the operator that substitutes some of the variables of $f(\bx)$ by 
field elements. It is easy to see that by substituting some variables by field elements, the new polynomial will also
be computed by a circuit of size less than $s$, and in general the new polynomial will also belong to the same class
as $f$.

In addition, we will need to discuss closure under three different operators:
\begin{itemize}
	\item {\bf Directional partial derivatives:} The partial derivatives $\frac{\D f}{\D x_i}$ of a polynomial $f$ are defined in the
	usual sense (over finite fields we use the formal definition for polynomials). We define the  {\em first order partial 
	derivative of $f$ in direction $\ba \in \F^n$} to be 
	$$f^{(1)}(\ba,\bx) \triangleq \ds\sum_{t=1}^{n} a_{t}  \cdot \frac{\D f}{\D x_{t}}(\bx)$$
	(see Definition~\ref{def:dirder}). Apart from the class of general circuits (and formulas) that are  closed 
	under taking first order derivatives \cite{BS83},  the class of sparse polynomials (depth $2$ circuits) is
	also closed under directional partial derivatives. Note, however, that depth-$3$ circuits with at most $k$ 
	multiplication gates, also known as $\Sigma\Pi\Sigma(k)$ circuits, are {\em not} closed under directional partial 
	derivatives as these might increase the top fanin.
	
	\item {\bf Homogeneous components:} If $f \in \F[\bx]$
	is a polynomial of degree $d$, we will denote the homogeneous component of degree $k$ of $f$ by $H^k(f(\bx))$. 
	General circuits and formulas are close under taking homogeneous components, and the same also holds for the 
	class of sparse polynomials (see e.g the proof of Lemma~\ref{lem:interp}).
	
	\item {\bf Shifts:} Here we require that a class will be closed under the operation $f(\bx) \mapsto f(\bx + \ba)$ for some
	$\ba \in \F^n$. Again, circuits and formulas as closed to shifts, however, the class of sparse polynomials is not.
\end{itemize}

We now describe our main result that solves the SET problem given a PIT algorithm.

\begin{thm}[Main theorem]\label{thm:main-int}
Let $\F$ be a field of characteristic zero. Let $\cM_1$ and $\cM_2$ be two circuit classes such that
\begin{enumerate}
	\item $\cM_1$ is closed under taking homogeneous components
	and closed under (first-order) directional derivatives.
	\item $\cM_2$ is closed under
	taking shifts.
	\item We have a (white-box) black-box PIT algorithm $\cal P$ for polynomials in
	$\cM_1, \cM_2$ and for polynomials of the form $f-g$, where $f \in \cM_1$ and $g \in \cM_2$.
\end{enumerate}
 Then,  there exists an algorithm $\cal S$ that, given (white-box) black-box
access to  polynomials $f \in \cM_1, g \in \cM_2$ and a bound $d$ on the their degree, returns $\ba \in \F^n$ so that
$g(\bx+\ba) \equiv f(\bx)$, if such a shift exists, or returns FAIL, if none exist.

Furthermore:
\begin{itemize}
	\item The running time of $\cal S$ is polynomial in the running time of $\cal P$ and in the other parameters ($n,d$).
	\item If the PIT algorithm $\cal P$ is deterministic then so is $\cal S$.
	\item All of the above holds also for the case when $\F$ is a finite field with characteristic greater than $d$.
\end{itemize}
\end{thm}

Combining Theorem~\ref{thm:main-int} with Lemma~\ref{lem:sz} we obtain  a randomized SET algorithm for any pair of 
polynomials.

\begin{thm}[Randomized SET for pairs of polynomials]\label{thm:main-rand-int}
Let $\F$ be a field of characteristic zero or of characteristic larger than $d$. There exists a randomized algorithm that, given
black box access to $f,g \in \F[\bx]$ of degree at most $d$, returns $\ba \in \F^n$
such that $g(\bx+\ba) \equiv f(\bx)$, if such a shift exists, or FAIL otherwise. 
The algorithm runs in time $\poly(n,d, \log(1/\veps))$, where $\eps$ is the  probability or returning a wrong answer (i.e., FAIL if a shift exists or a shift if none exists).
\end{thm}

\begin{remark} An interesting fact about Theorem~\ref{thm:main-rand-int} is that the algorithm we obtain has a two
sided error (this can be seen from the proof). This fact is in contrast to the fact that most randomized algorithms in the
algebraic setting have one-sided error.
\end{remark}

Theorem~\ref{thm:main-int} already leads to deterministic algorithms for certain restricted models. For instance, in the recent
works of Forbes and Shpilka \cite{ForbesShpilka12,ForbesShpilka12a} and of Forbes, Saptharishi and Shpilka
\cite{ForbesSS13}, the authors obtain a quasi-polynomial deterministic PIT algorithm
for read-once oblivious algebraic branching programs (ROABPs). Their result, together with our algorithm, imply that
we can find out whether two ROABPs are shift-equivalent in deterministic quasi-polynomial time. Since this class also captures
tensors,\footnote{We note that the work \cite{AgrawalSS13} also gives a black-box PIT algorithm for tensors.} an application 
of our result is that we can find out whether two tensors are shift-equivalent in quasi-polynomial time
(we refer the reader to \cite{ForbesShpilka12a,ForbesSS13} for definitions of ROABPs and tensors).

\begin{cor}\label{ROABP=shift}
There is a deterministic quasi-polynomial time algorithm that given black-box access to two polynomials $f$ and $g$ computed by
read-once oblivious algebraic branching programs, decides whether there exists $\ba\in\F^n$ such that
$f(\bx+\ba) \equiv g(\bx)$ and in case that such a shift exists, the algorithm outputs one.
\end{cor}

As the class of sparse polynomials is closed under taking homogeneous components and under first order directional
derivatives (a directional derivative blows up the size of the circuit by at most a factor of $n$) we obtain the
following corollary.

\begin{cor}\label{cor:sparse-shift}
Let $\cM_2$ be any circuit class so that
\begin{enumerate}
	\item $\cM_2$ is close under shifts.
	\item There is a deterministic PIT algorithm testing if $f-g$ is zero for sparse $f$ and $g \in \cM_2$.
\end{enumerate}
Then, we can test whether $f$ and $g$ are shift-equivalent deterministically in time $\poly(n,s)$
\end{cor}

As an application of our main theorem in the white-box model, we note that Saha et al. gave a polynomial time algorithm
for testing whether a given sparse polynomial equals a $\Sigma\Pi\Sigma(k)$ circuit \cite{SahaSS11}. Since their algorithm
works in the white-box model, we can utilize it in the variant of our main theorem in the white-box model to find whether
a given sparse polynomial and a polynomial in $\Sigma\Pi\Sigma(k)$ are shift-equivalent.
We also note that we can  make their algorithm work in the black-box case as well. Using the reconstruction algorithms of 
\cite{Shpilka09,KarninShpilka09} we can first reconstruct the $\Sigma\Pi\Sigma(k)$ circuit in quasi-polynomial time. We can 
also interpolate the sparse polynomial in polynomial time (for interpolation of sparse polynomials see e.g. 
\cite{KlivansSpielman01}) and then apply our methods together with the PIT algorithm of Saha et al. to solve the 
shift-equivalence problem.\footnote{Note that the reconstruction algorithm of \cite{Shpilka09,KarninShpilka09} returns 
so-called generalized $\Sigma\Pi\Sigma(k)$ circuits. Nevertheless, one can observe that the algorithm of Saha et al.  
works for such circuits as well.}

%%%%%%%%%%%%%%%%%%%%%%%%%%%%%%%%%%%%%%%%%%%%%%%%
\subsection{Overview of the algorithm}\label{sec:proof overview}

In this section we give a short overview our algorithm and its analysis. Assume we are given $f(\bx)$ and
$g(\bx)$ and we have to find $\ba\in\F^n$ such that $f(\bx+\ba)=g(\bx)$. Let us assume w.l.o.g. that $\deg(f)=\deg(g)=d$.
Let us also denote $f(\bx)=\sum_{i=0}^{d}H^i(f(\bx))$ where each $H^i(f)$ is homogeneous of degree $i$ and similarly,
$g(\bx)=\sum_{i=0}^{d}H^i(g(\bx))$.

Now, let us compute the homogeneous components of $f(\bx+\ba)$. Denote with $H^i(f(\bx+\ba))$ the homogeneous
part of degree $i$ of $f(\bx+\ba)$. We have that
$$H^d(f(\bx+\ba)) = H^d(f(\bx)).$$
Thus, our first step of the algorithm is to verify that
$$H^d(g(\bx))=H^d(f(\bx)).$$
Next we move to degree $d-1$. A quick calculation gives
\begin{equation}\label{eq:one}
H^{d-1}(g(\bx))=H^{d-1}(f(\bx+\ba)) = H^{d-1}(f(\bx)) + \sum_{k=1}^{n}a_k\cdot \frac{\partial H^d(f(\bx))}{\partial x_k}.
\end{equation}
Observe that this is a linear equation in the entries of $\ba$. It turns out that if our circuit class is closed under directional
derivatives, that is, if the polynomial $\sum_{k=1}^{n}a_k\cdot \frac{\partial H^d(f(\bx))}{\partial x_k}$ belongs to the same 
circuit class as $f(\bx)$ (or a slightly larger class), and if we have a hitting set for the class of polynomials of the form
$\sum_{k=1}^{n}a_k\cdot \frac{\partial H^d(f(\bx))}{\partial x_k}$, for every $\ba \in \F^n$, 
then we can solve this system of equations and find some solution $\bb$ such that 
$$H^{d-1}(g(\bx))=H^{d-1}(f(\bx+\bb)) = H^{d-1}(f(\bx)) + \sum_{k=1}^{n}b_k\cdot \frac{\partial H^d(f(\bx))}{\partial x_k}.$$
As we will see in section~\ref{sec:poleq-bb}, if we allow randomness then we can also solve this system of equations without 
having a hitting set.

Note that at this point we might have $\bb\neq \ba$. Hence, we have found a shift $\bb$ that makes the homogeneous
parts of degree $d$ and $d-1$ in $f$ and $g$ equal. We now consider the homogeneous component of degree $d-2$. Here we
have the system of equations
\begin{equation}\label{eq:int-deg-2}
H^{d-2}(g(\bx)) =H^{d-2}(f(\bx+\ba)) = H^{d-2}(f(\bx)) + \sum_{k=1}^{n}a_k\cdot \frac{\partial H^{d-1}(f(\bx))}{\partial x_k} +
\sum_{\ell,k=1}^{n}a_\ell a_k\frac{\partial^2 H^d(f(\bx))}{\partial x_\ell\partial x_k}.
\end{equation}
And now we seem to be in trouble as this is a system of quadratic equations in the entries of $\ba$.
Here comes our crucial observation. Recall that we have found $\bb$ such that
$$H^{d-1}(g(\bx)) = H^{d-1}(f(\bx)) + \sum_{k=1}^{n}b_k\cdot \frac{\partial H^d(f(\bx))}{\partial x_k}.$$
We also have that
$$H^{d-1}(g(\bx)) =H^{d-1}(f(\bx+\ba))= H^{d-1}(f(\bx)) + \sum_{k=1}^{n}a_k\cdot \frac{\partial H^d(f(\bx))}{\partial x_k}.$$
Hence,
$$\sum_{k=1}^{n}(a_k-b_k)\cdot \frac{\partial H^d(f(\bx))}{\partial x_k}=0.$$
This means that the directional derivative of $H^d(f(\bx))$ in direction $\ba-\bb$ is zero.
Or, in other words, that the polynomial $H^d(f(\bx))$ is fixed along that direction. This means that no matter how
many derivatives we take along direction $\ba-\bb$ we always
get the zero polynomial. Therefore, if we take a second derivative in direction $\bc$ or in direction
$\bc + (\ba-\bb)$ we will get the same answer, no matter what $\bc$ is. In particular, this gives
$$\sum_{\ell,k=1}^{n}a_\ell a_k\frac{\partial^2 H^d(f(\bx))}{\partial x_\ell\partial x_k} = \sum_{\ell,k=1}^{n}b_\ell b_k
\frac{\partial^2 H^d(f(\bx))}{\partial x_\ell\partial x_k},$$ as both sides compute the second directional derivatives in directions
$\ba$ and $\bb$, respectively. Going back we now have that system~\eqref{eq:int-deg-2}
is equivalent to the system
\begin{equation}\label{eq:int-deg-2-b}
H^{d-2}(g(\bx)) =H^{d-2}(f(\bx+\ba)) = H^{d-2}(f(\bx)) + \sum_{k=1}^{n}a_k\cdot \frac{\partial H^{d-1}(f(\bx))}{\partial x_k} +
\sum_{\ell,k=1}^{n}b_\ell b_k\frac{\partial^2 H^d(f(\bx))}{\partial x_\ell\partial x_k}.
\end{equation}
Since we already computed $\bb$, we can look for a solution to both systems of equations \eqref{eq:int-deg-2}
and \eqref{eq:int-deg-2-b} (as linear systems in the coefficients of $\ba$). Once we find such a solution, say $\bc$, we can
use it to set up a new system of equations involving the homogeneous components of degree $d-3$ and so on.

Thus, our algorithm works in iterations. We start by solving a system of linear equations. We then use the solution that we
found to set up another system and then we find a common solution to both systems. We use the solution that we have
found to construct a third system of equations and then solve all three systems together etc.
At the end we have a solution for all systems, and at this point it is not difficult to verify, that if such a shift $\ba$ exists, then
the solution that we found is indeed a valid shift. This can be verified by running one PIT for checking whether the shift of
$f$ that we have found and $g$ are equivalent.

All the steps above can be completed using randomness, including solving the black-box system of equations, or using PIT
for the relevant circuit classes.

%%%%%%%%%%%%%%%%%%%%%%%%%%%%%%%%%%%%%%
\subsection{Related work}

The works of Grigoriev,  Lakshman, Saunders and Karpinski \cite{grig-interp, grig-sparseshift, lak-shiftsuni}, try to
solve the problem of finding sparse shifts of given polynomials, in order to make their representation more succinct.
In \cite{grig-interp}, Grigoriev and Karpinski studied the problem of finding sparse affine-shifts of multivariate polynomials
$f(\bx)$, that is, transformations of the form $\bx \mapsto A \bx + \bb$ where $A$ is full-rank,
which make the input polynomial $f(A\bx+\bb)$ sparse. In \cite{lak-shiftsuni}, the authors consider the problem of finding 
sparse shifts of univariate polynomials, and of
determining uniqueness of a sparse shift. Given an input polynomial $f(x)$, they use a criterion based on the vanishing of the
Wronskian of some carefully designed polynomials, which depend on the derivatives $f^{(i)}(x)$, in order to obtain an efficient
algorithm for the univariate case.

Later, in \cite{grig-q2tcs}, Grigoriev gave three algorithms for the SET problem, which were
polynomial in the size of the dense representation of the input polynomials. His algorithms were based on a structural
result about the set of shifts that stabilize the polynomial, that is, the set of points $\ba \in \F^n$ for which
$f(\bx+\ba) \equiv f(\bx)$. We denote this stabilizer by $\cS_{f}$. He noticed that $\cS_f$ is a subspace of
$\F^n$ and that the set of shifts that are solutions
to the SET problem with input polynomials $f,g$, which we denote $\cS_{f,g}$, is a coset of $\cS_f$.
After this observation, Grigoriev established the following recursive relations between $\cS_{f,g}$ and
$\cS_{\frac{\D f}{\D x_i}, \frac{\D g}{\D x_i}}$, for each $x_i$:
\[  \cS_{f,g} = \bigcap_{i=1}^n \cS_{\frac{\D f}{\D x_i}, \frac{\D g}{\D x_i}} \cap \{ \ba \in \F^n \ : \ f(\ba) = g(0) \}. \]
From these relations, Grigoriev devised a recursive algorithm that finds $\cS_{f,g}$ by finding the subspaces corresponding
to $\cS_{\frac{\D f}{\D x_i}, \frac{\D g}{\D x_i}}$. Because this recursive procedure will find all the subspaces
$\cS_{\frac{\D f}{\D m}, \frac{\D g}{\D m}}$ for every monomial $m$ of degree less than or equal to $d = \max(d_f, d_g)$,
the running time of his algorithm is bounded by $n^{O(d)}$. Our approach is different from Grigoriev's in the sense that 
we avoid the recursive relations and find a shift by iteratively constructing a shift which makes $f$ and $g$ agree on 
their homogeneous parts of up to a certain degree, starting from the homogeneous parts of highest degree down to
the homogeneous parts of lowest degree (i.e., the constant term).

The study of equivalences of general polynomials under affine transformations, which we refer to as
affine-equivalence, was started by Kayal in \cite{kayal12} (note that this generalizes the problem studied in 
\cite{grig-interp}). We say that $f$ and $g$ are affine-equivalent if there exists
a matrix $A$ and a shift $\bb$ such that $f(\bx)=g(A\bx+\bb)$. In this  work, Kayal analyzes whether a given
polynomial $f$ can be obtained by an affine transformation of a given polynomial $g$, where $g$ is usually taken to
be a ``complete'' polynomial in some arithmetic circuit class, such as the Determinant or Permanent polynomials. In his
paper, Kayal establishes NP-hardness of the general problem of determining affine-equivalence between two
arbitrary polynomials. Moreover, he provides randomized algorithms for the affine-equivalence problem when
one of the polynomials is the Permanent or the Determinant and the affine transformations $\bx \mapsto A\bx + \bb$
are of a special form (in the case of Determinant and Permanent, the matrix $A$ must be invertible). Kayal provides
randomized algorithms for some other classes of homogeneous polynomials, and for more details we refer the reader
to the paper \cite{kayal12}. Our work is different from Kayal's work since in our setting we are only interested in
shift-equivalences, and in this feature we are less general than Kayal's work, but we also consider larger classes of
polynomials, in which case we are more general than Kayal's work.

Following the initial publication of this manuscript, an anonymous reader pointed out an alternative way to solve the SET problem using a randomized algorithm. This approach uses a lemma due to Carlini \cite{car06} (see also
\cite[Lemma 17]{kayal12}) and an argument implicit in Kayal's work \cite[Section 7.3]{kayal12}. We now discuss and compare this alternative approach to ours.

 In his lemma, Carlini uses a linear transformation on the variables in order to get
rid of ``redundant variables," that is, variables $x_i$ for which (after a suitable change of basis) the derivative 
$\frac{\D f}{\D x_i}$ of the polynomial is zero. The idea is to use this lemma to eliminate the
``redundant variables" and work only with the ``essential variables." Once we find such a linear transformation, 
one can solve Equation~\eqref{eq:one}
(there will be at most one solution, since there are no more redundant variables). Then, we reduce the original problem to another SET problem 
on lower degree polynomials by subtracting the homogeneous part of largest degree. We give the details (which do not appear elsewhere in the literature) in Appendix~\ref{sec:altmain}.
In a sense, this approach is almost identical to ours. In the first step of our algorithm we solve Equation~\eqref{eq:one}
and get an affine subspace as solution. This affine subspace can be thought of as being composed of 
the space of all assignments to the ``non-essential'' variables shifted by the unique solution. Then, in the next step,
we prune this space further according to the essential variables of the degree $d-1$ part etc.
The advantages of our approach come in when trying to de-randomize SET using  deterministic PIT for restricted classes.
When following Carlini's lemma and reducing the number of variables, one  needs to solve PIT for the composition of the original circuits with a linear transformation. This is not
necessary in our approach, which has weaker PIT requirements. While some circuit classes are closed under linear transformations, this is not the case in 
general. For example, the class of sparse polynomials is not closed under linear transformations. Thus, one will not be 
able to deduce polynomial time algorithms to certain instances of SET like those that follow from our approach 
(see Corollary~\ref{cor:sparse-shift} and the discussion following it).

Another issue with the algorithm obtained from Carlini's lemma is that in each step 
of the recursion we need to subtract an affine shift of the homogeneous component of maximal degree from 
each of the polynomials. Thus, we need PIT for classes that are closed under linear combinations of polynomials from the 
class. However, some restricted circuit classes do not satisfy such closure properties. For example, when executing this 
algorithm on depth-$3$ circuits with bounded top fan-in, we may get, at some step of the algorithm,  a depth-$3$ circuit 
with unbounded top fan-in and so we will not be able to use current deterministic algorithms.

Another line of works that has some resemblance to our results is the study of black-box groups. The well-known
algorithm of Sims (see the book \cite{Seress2003}) finds a small set of generators for a permutation group given by
black-box access. Our algorithm can be seen as finding a basis for the affine space of all shifts from $f$ to $g$ so in
that sense it also finds generators for a black-box group where we do not have direct access to the group but rather to
the objects it acts upon. An interesting point is that while Sims' algorithm works by constructing the group in a ``bottom to
top'' fashion, namely starting with the identity element and slowly finding more generators, we on the other hand find
a sequence of affine spaces, each contained in the proceeding ones until we reach the final space.

%%%%%%%%%%%%%%%%%%%%%%%%%%%%%%%%%%%%%%
\subsection{Organization}

The rest of the paper is organized as follows: in section~\ref{sec:prelim} we introduce some useful lemmas that one obtains from
having PIT for a class of circuits. In section~\ref{sec:homshifts} we introduce some properties of homogeneous components of
shifts of polynomials. In section~\ref{sec:kershifts} we define the space of shifts of a polynomial that do
not change the polynomial at all (i.e. the stabilizer) and describe some of its properties. In section~\ref{sec:main} we formally state and
prove the main theorem of this paper, describing and analyzing the algorithms for testing shift-equivalence.

%%%%%%%%%%%%%%%%%%%%%%%%%%%%%%%%%%%%%%%%%%%%%%%%%%%%
%%%%%%%%%%%%%%%%%%%%%%%%%%%%%%%%%%%%%%%%%%%%%%%%%%%%
%%%%%%%%%%%%%%%%%%%

\section{Preliminaries}\label{sec:prelim}

In this section, we establish some notation that will be used throughout the paper and introduce some useful lemmas
about simulation of circuits in the black-box setting. In addition, we state and prove a lemma on how to solve a linear
system of polynomial equations in the black-box (or white-box) setting, given that one has a 
black-box (or white-box) PIT algorithm for linear combinations of the polynomials in question. 
In section~\ref{sec:poleq-wb}, we show that if we are given white-box access to the input polynomials, then 
white-box PIT for linear combinations of these polynomials is enough to solve linear system of equations with these 
polynomials. On the other hand, in section~\ref{sec:poleq-bb}, if we are given black-box access to the input polynomials, 
then we show that having a hitting set is enough. Notice that although the result in section~\ref{sec:poleq-bb} seems to 
be stronger than the result in section~\ref{sec:poleq-wb}, the two results are actually not comparable, since in
section~\ref{sec:poleq-bb} we are assuming that we have a hitting set, which is a stronger assumption than only
having a white-box PIT algorithm, which is the assumption in section~\ref{sec:poleq-wb}.

From this point on, we will use boldface for vectors, and regular font for scalars. Thus, we will denote the vector
$(x_1, \ldots, x_n)$ by $\bx$ and if we want to multiply the vector
$\bx$ by a scalar $z$ we will denote this product by $z \bx$.

We will also assume that the ground field $\F$ either has characteristic zero or that its characteristic is larger than
the degree of any polynomial that we will be working with. This assumption will be crucially used throughout
sections~\ref{sec:homshifts} and \ref{sec:kershifts}. In addition, we denote the characteristic of $\F$ by $\char(\F)$.

%%%%%%%%%%%%%%%%%%%%%%%%%%%%%%%%%%%%%%%%%%%%%%%%
\subsection{Interpolation in the Black-Box setting}

For many problems in algebraic computation, it is useful to work with the homogeneous components of a polynomial, instead
of directly working with the entire polynomial. In the black-box setting, we do not have direct black-box access to the
homogeneous components of the given polynomial $f$. However, the next lemma shows that from black-box access to $f$ we
can obtain black-box access to its homogeneous components $H^0(f), \ldots, H^d(f)$.

\begin{lem}\label{lem:interp}
If we are given black-box access to a circuit $C(\bx)$ that computes a polynomial $f(\bx) \in \F[\bx]$ of degree
$d$, then we can obtain black-box access to the homogeneous components of $f$.
\end{lem}
\begin{proof}
We know that $f(\bx) = \ds\sum_{i=0}^d H^i(f(\bx))$.
Hence, we have that $f(z \bx) = \ds\sum_{i=0}^d z^i H^i(f(\bx))$. If we let $\{\alpha_i\}_{0 \le i \le d}$ be
$d+1$ distinct elements of $\F$ (or of an extension field of $\F$) and if we evaluate $C$ on the points $\alpha_i \bx$
we obtain the following equality:

\[
\begin{pmatrix}
1 & \alpha_0 & \alpha_0^2 & \ldots & \alpha_0^d \\
1 & \alpha_1 & \alpha_1^2 & \ldots & \alpha_1^d \\
1 & \alpha_2 & \alpha_2^2 & \ldots & \alpha_2^d \\
\vdots & \vdots & \vdots & \vdots & \vdots \\
1 & \alpha_d & \alpha_d^2 & \ldots & \alpha_d^d
\end{pmatrix} \cdot
\begin{pmatrix}
H^0(f(\bx)) \\ H^1(f(\bx)) \\ H^2(f(\bx)) \\ \vdots \\ H^d(f(\bx))
\end{pmatrix} =
\begin{pmatrix}
f(\alpha_0 \bx) \\ f(\alpha_1 \bx) \\ f(\alpha_2 \bx) \\ \vdots \\ f(\alpha_d \bx)
\end{pmatrix}\]

The matrix on the left side is a Vandermonde matrix, which is known to be invertible. Hence, by left-multiplying by
its inverse we obtain:

\[
\begin{pmatrix}
1 & \alpha_0 & \alpha_0^2 & \ldots & \alpha_0^d \\
1 & \alpha_1 & \alpha_1^2 & \ldots & \alpha_1^d \\
1 & \alpha_2 & \alpha_2^2 & \ldots & \alpha_2^d \\
\vdots & \vdots & \vdots & \vdots & \vdots \\
1 & \alpha_d & \alpha_d^2 & \ldots & \alpha_d^d
\end{pmatrix}^{-1} \cdot \begin{pmatrix}
f(\alpha_0 \bx) \\ f(\alpha_1 \bx) \\ f(\alpha_2 \bx) \\ \vdots \\ f(\alpha_d \bx)
\end{pmatrix} =
\begin{pmatrix}
H^0(f(\bx)) \\ H^1(f(\bx)) \\ H^2(f(\bx)) \\ \vdots \\ H^d(f(\bx))
\end{pmatrix}
\]

Since we have black-box access to the values $f(\alpha_i \bx)$ through the circuit $C$, we
also have black-box access to the homogeneous components of $f$ through this construction.
\end{proof}

%%%%%%%%%%%%%%%%%%%%%%%%%%%%%%%%%%%%%%%%%%%%%%%%
\subsection{Finding linear dependencies among polynomials in the White-Box Setting}\label{sec:poleq-wb}

Suppose we have explicit access to the circuits computing the polynomials $g, h_1, h_2, \ldots, h_k \in \F[\bx]$. Then, how 
can we decide whether $g$ is in the linear span of $h_1, h_2, \ldots, h_k$? That is, does there exist an
$\ba \in \F^k$ such that  $g(\bx) \equiv \sum_{i=1}^k a_i h_i(\bx)$? Notice that we cannot try to solve a linear system for 
each possible monomial of $g$, since this process might lead us to an exponential number of equations.

In this subsection we answer the question above, assuming that we have a white-box PIT algorithm that hits the 
$\F$-span of the polynomials $g, h_1, \ldots, h_k$, that is, polynomials of the form 
$a_0 g(\bx) + \sum_{i=1}^k a_i h_i(\bx)$, where $a_i \in \F, \ 0 \le i \le k$. Moreover, we can find such a linear combination, 
if it exists.

\begin{lem}[Decision to search reduction for white-box PIT]\label{lem:PIT-decision-to-search}
Given an arithmetic circuit $C$ computing a non-zero $n$-variate polynomial $f$ of degree $d$, and a white-box
deterministic PIT algorithm that runs in polynomial time, we can find, in deterministic
polynomial time, a point $\ba\in\F^n$ such that $f(\ba)\neq 0$.
\end{lem}

\begin{proof}
Let $S=\{a_0,\ldots,a_d\}$ be a set of $d+1$ distinct values from $\F$. Notice that we can check, using the PIT algorithm,
whether the restriction $x_1 = a_i \in S$ makes $f$ vanish. Since the degree of
$f$ is $d$ and $f\not\equiv 0$, there exists a value of $a_i \in S$ such that $f(a_i, x_2, \ldots, x_n) \not\equiv 0$. Hence, by
a linear scan over $S$ we can find such an index $0 \le i \le d$ such that $f(a_i, x_2, \ldots, x_n) \not\equiv 0$.
Fix $x_1 = a_i$ and repeat this procedure with the other variables $\{x_2, \ldots, x_n\}$.
The running time is clearly bounded by $nd$ times the running time of the PIT algorithm.
\end{proof}

\begin{lem}\label{lem:poleq-wb}
Suppose we are given circuits computing the polynomials $g, h_1, h_2, \ldots, h_k \in \F[\bx]$. Assume further that
we have a deterministic white-box PIT algorithm for linear combinations of  $g, h_1, h_2, \ldots, h_k$.
Then, there exists a deterministic algorithm, with running time polynomial in the sizes of the circuits and $k$,
which decides whether there exists $\ba \in {\F}^k$ such that $\sum_{i=1}^k a_i h_i(\bx)\equiv g(\bx)$.
Moreover, the algorithm will output such an $\ba$, if there exists one.
\end{lem}
\begin{proof}
We will be relying on the decision-to-search reduction of Lemma~\ref{lem:PIT-decision-to-search} and we will use it
implicitly throughout in the proof.

As a first step, find a point $\ba_1$ such that $h_1(\ba_1)\neq 0$. Next, find a point $\ba_2$ such that
$h_1(\ba_2)h_2(\ba_1)-h_1(\ba_1)h_2(\ba_2) \neq 0$.  If no such point $\ba_2$ exists then we can discard $h_2$,
since in this case $h_2$ will be in the span of $h_1$. If $h_1(\bx)h_2(\ba_1)-h_1(\ba_1)h_2(\bx) \not\equiv 0$, then
by Lemma~\ref{lem:PIT-decision-to-search} we can find such $\ba_2$. That is why we can discard $h_2$ in case we are not
able to find such a point. Proceed in this manner until we have scanned through all $h_1,\ldots, h_k$. More accurately, 
assume (wlog) that the polynomials $h_1, \ldots, h_{c}$, for some $c<\ell$, are linearly independent and their span 
contains the polynomials $h_1,\ldots,h_{\ell-1}$. At the $\ell^{th}$ step
we consider the $c+1 \times c+1$ matrix $M_\ell$ that is defined as follows:  
$M_\ell[i,j] = \begin{cases} h_i(a_j), \text{ if } j \leq c \\ h_i(\bx), \text{ if } j = c+1\end{cases}$. We
find a point $\ba_\ell$ for which the determinant of $M_\ell$ is non-zero.
Notice that this determinant is merely a linear combination of the polynomials $h_1, \ldots, h_\ell$, hence we have PIT for this 
polynomial and we can find such point $\ba_\ell$, if one exists. 

W.l.o.g., we can assume that $h_1, \ldots, h_r$ form a basis for the space defined by the $\F$-span of the
polynomials $h_1, \ldots, h_k$. Hence, by our linear scan through the $h_i$'s, we have found $\ba_1,\ldots,\ba_r$ such
that the $r\times r$ matrix $H$ having in its $(i,j)^{th}$ entry the value $h_i(\ba_j)$ is full rank.

We now evaluate the polynomial $g$ on all those $r$ points and find the unique linear combination yielding
$\sum_{i=1}^{r}b_i h_i(\ba_j) = g(\ba_j)$ for $1\leq j\leq r$. Notice that we can find $\bb \in \F^r$ by solving a system of
linear equations over $\F$. This vector $\bb$ will be unique since $H$ is full rank. Notice that, by uniqueness of $\bb$
and by the fact that $h_1, \ldots, h_r$ form a basis
for the linear span of the $h_i$'s, we have that $g$ is a linear combination of the $h_i$'s if, and only if,
$\sum_{i=1}^{r}b_i h_i(\bx) \equiv g(\bx)$. Hence, all we need to do is to check whether
$\sum_{i=1}^{r}b_i h_i(\bx) \equiv g(\bx)$. We can test this polynomial equality by running our PIT algorithm on the
polynomial $g(\bx) - \sum_{i=1}^{r}b_i h_i(\bx)$. If the PIT algorithm returns that this polynomial is the zero polynomial,
then we found a linear combination. Otherwise, the algorithm returns that there exists no linear combination.
\end{proof}

%%%%%%%%%%%%%%%%%%%%%%%%%%%%%%%%%%%%%%%%%%%%%%%%
\subsection{Finding linear dependencies among polynomials in the Black-Box Setting}\label{sec:poleq-bb}

Here we assume that we only have black-box access to polynomials $g, h_1, h_2, \ldots, h_k \in \F[\bx]$ and we wish 
to solve the 
same question as the one posed in the previous subsection, assuming a black-box PIT. We first note that the proof of 
Lemma~\ref{lem:poleq-wb} also works in the black-box case, but since we have a stronger assumption, namely, a hitting set 
rather than a white-box PIT algorithm, we have a more direct solution: We shall find a set of points $S$ for which the equation
$g(\bx) \equiv \sum_{i=1}^k a_i h_i(\bx)$ is true if, and only if, $g(\bc) \equiv \sum_{i=1}^k a_i h_i(\bc)$ for every $\bc \in S$.

It turns out that if we have a hitting set $\cH$ that hits the $\F$-span of the polynomials
$g, h_1, \ldots, h_k$, then the points of $\cH$ give the required set $S$.

The following lemma states formally the answer to the question above:

\begin{lem}\label{lem:poleq-bb}
Suppose we have black-box access to polynomials $g, h_1, h_2, \ldots, h_k \in \F[\bx]$ and that we
have a hitting set $\cH$ that hits the $\F$-span of the polynomials $g, h_1, h_2, \ldots, h_k$. Then,
there exists a deterministic algorithm, with running time polynomial in $|\cH|$ and $k$, which decides whether there
exists $\ba \in \F^k$ such that $g(x) \equiv \sum_{i=1}^k a_i h_i(\bx)$. Moreover, the algorithm will output
such an $\ba$, if one exists.
\end{lem}

\begin{proof}
Let $s = |\cH|$ and let $\bc_1, \bc_2, \ldots, \bc_s$ be an arbitrary ordering of the elements of $\cH$.
For a polynomial $f \in \F[\bx]$, define the vector $\bv_f \in \F^s$ as follows:
$\bv_f = (f(\bc_1), f(\bc_2), \ldots, f(\bc_s))^T.$ Then, it is enough to prove
the following equivalence: $g(\bx) \equiv \sum_{i=1}^k a_i h_i(\bx)$ if, and only if, $\bv_g = \sum_{i=1}^k a_i \bv_{h_i}$.
This implies the lemma, since given the polynomials $g, h_1, \ldots, h_k$ and $\cH$, we can construct the
vectors $\bv_g, \bv_{h_1}, \ldots, \bv_{h_k}$ and just solve the system of linear equations
$\bv_g = \sum_{i=1}^k a_i \bv_{h_i}$, where the $a_i$'s are the unknowns.

Here is the proof of the equivalence above: $g(\bx) \equiv \sum_{i=1}^k a_i h_i(\bx)$ implies that
$g(\bc_r) = \sum_{i=1}^k a_i h_i(\bc_r)$ for all $r \in [s]$, which implies that $\bv_g = \sum_{i=1}^k a_i \bv_{h_i}$.
On the other hand, if $\bv_g = \sum_{i=1}^k a_i \bv_{h_i}$, then we have $\bv_g - \sum_{i=1}^k a_i \bv_{h_i} = 0$, which
implies that $g(\bc_r) - \sum_{i=1}^k a_i h_i(\bc_r) = 0$, for all $r \in [s]$. Since $\cH$ hits linear combinations of
$g, h_1, \ldots, h_k$, the last set of equalities implies that the polynomial $g(\bx) - \sum_{i=1}^k a_i h_i(\bx)$ vanishes
on all points of $\cH$, and therefore it must be the zero polynomial. This implies that
$g(\bx) \equiv \sum_{i=1}^k a_i h_i(\bx)$ and proves the lemma.
\end{proof}

Now, what if we do not have such a hitting set $\cH$, but we are allowed randomness? Then, we can still 
answer the question above in the positive, with high probability, and find such a linear combination if one exists. 
More formally, we have:

\begin{lem}\label{lem:poleq-rand}
Suppose we have black-box access to polynomials $g, h_1, h_2, \ldots, h_k \in \F[\bx]$ and an upper bound $d$ on their
degrees.  Let $0 < \veps < 1$. Then,
there exists a randomized algorithm, with running time $\poly(d, \log(1/\veps), k)$, which decides correctly with
probability at least $1-\veps$ whether there exists $\ba \in \F^k$ such that $g(x) \equiv \sum_{i=1}^k a_i h_i(\bx)$. 
Moreover, with the same error probability the algorithm will output such an $\ba$, if one exists.
\end{lem}

\begin{proof}
The proof of this lemma is similar to the proof of the white-box case, the difference being in the fact that we will choose our
evaluation points according to Lemma~\ref{lem:sz}. Let $S$ be a set of size $|S| = \lfloor \frac{2dk}{\veps} \rfloor $. 
  
As a first step, pick a point $\ba_1$ at random from $S^n$. By Lemma~\ref{lem:sz}, if $h_1 \not\equiv 0$ then 
$h_1(\ba_1) = 0$ with probability $\le \veps/2k$. If $h_1(\ba_1) = 0$ but $h_1 \not\equiv 0$, then we will just 
assume that $h_1 \equiv 0$ and we will discard it (and in this part that our algorithm may make a mistake). 
Next, pick a point $\ba_2$ at random from $S^n$. Again, by Lemma~\ref{lem:sz}, if
$h_1(\ba_1)h_2(\bx)-h_1(\bx)h_2(\ba_1) \not\equiv 0$ then $h_1(\ba_2)h_2(\ba_1)-h_1(\ba_1)h_2(\ba_2) = 0$ 
with probability $\le \veps/2k$. If $h_1(\ba_2)h_2(\ba_1)-h_1(\ba_1)h_2(\ba_2) = 0$ we will always assume that $h_2$ is in 
the span of $h_1$ and thereby we will discard $h_2$ (in this part our algorithm may again make a mistake).
We thus proceed in this manner, following the footsteps of the proof of Lemma~\ref{lem:poleq-wb}.

As before we assume (wlog) that $h_1, \ldots, h_r$ form a basis for the space defined by the $\F$-span of the
polynomials $h_1, \ldots, h_k$. Hence, by our linear scan through the $h_i$'s, we have found $\ba_1,\ldots,\ba_r$ such
that the $r\times r$ matrix $H$ having in its $(i,j)^{th}$ entry the value $h_i(\ba_j)$ is full rank. The probability that we made 
a mistake until this point will be $\le r \veps/2k \le \veps/2$, by the union bound.

We continue as in the proof of Lemma~\ref{lem:poleq-wb}. Assuming that we made no mistake so far, we can find (by 
solving linear equations over $\F$) the unique point $\bb$ such that $g$ is a linear combination of the $h_i$'s if, and only if,
$\sum_{i=1}^{r}b_i h_i(\bx) \equiv g(\bx)$. Hence, all we need to do is to check whether
$\sum_{i=1}^{r}b_i h_i(\bx) \equiv g(\bx)$. We can test this polynomial equality by applying Lemma~\ref{lem:sz} on the
polynomial $g(\bx) - \sum_{i=1}^{r}b_i h_i(\bx)$, again drawing the point at random from $S^n$. 
If the PIT algorithm returns that this polynomial is the zero polynomial, then we found a linear combination. 
Otherwise, the algorithm returns that there exists no linear combination. The probability of the PIT making a mistake at this
step is $\le \veps/2k \le \veps/2$. Hence, the total error of the entire algorithm is bounded by $\veps/2 + \veps/2 = \veps$, as
claimed.
\end{proof}

%%%%%%%%%%%%%%%%%%%%%%%%%%%%%%%%%%%%%%%%%%%%%%%%
%%%%%%%%%%%%%%%%%%%%%%%%%%%%%%%%%%%%%%%%%%%%%%%%
%%%%%%%%%%%%%%%%%%%%%%%%%%%%%%%%%%%%%%%%%%%%%%%%
\section{Homogeneous Components of Shifts of a Polynomial}\label{sec:homshifts}

In this section we describe some properties of the homogeneous components of a shift of a polynomial.
Throughout this section, let $f(\bx) \in \F[\bx]$ be a polynomial of degree $d$, $\ba \in \F^n$ be a point
and $f(\bx+\ba)$ be a shift of $f$.
In general, when the field $\F$ is such that $\char(\F) = 0$ or $\char(\F) > d$,
the homogeneous components of $f(\bx+\ba)$ can be expressed
as a linear combination of the appropriate (formal) directional derivatives of the homogeneous components of $f$
on the direction $\ba$. Before we state these properties more formally, we will need the following definitions:

\begin{define}[Directional Derivatives]\label{def:dirder}
The (formal) directional derivative of $f(\bx) \in \F[\bx]$ of order $1$
on the direction $\ba$ is given by the following formula:
\begin{align}\label{eq:dirder-1}
f^{(1)}(\ba,\bx) \triangleq \ds\sum_{t=1}^{n}
a_{t}  \cdot \frac{\D f}{\D x_{t}}(\bx).
\end{align}
More, generally, The (formal) directional derivative of $f(\bx) \in \F[\bx]$ of order
$r$ on the
direction $\ba$ is given by the following formula:
\begin{align}\label{eq:dirder}
f^{(r)}(\ba,\bx) \triangleq \ds\sum_{\be \in [n]^r}
\left( \prod_{k=1}^r a_{e_k} \right) \cdot \frac{\D^{r} f}{\D x_{e_1} \ldots \D x_{e_r}}(\bx)
\end{align}
where we define $f^{(0)}(\ba, \bx) \triangleq f(\bx)$. If $f$ is not homogeneous, for each homogeneous component
$H^\ell(f)$ of $f$ we define:
\begin{align}
f_\ell^{(r)}(\ba,\bx) \triangleq \ds\sum_{\be \in [n]^r}
\left( \prod_{k=1}^r a_{e_k} \right) \cdot \frac{\D^{r} H^\ell(f)}{\D x_{e_1} \ldots \D x_{e_r}}(\bx).
\end{align}
\end{define}

Note that equation~\eqref{eq:dirder} in definition~\ref{def:dirder} agrees with the usual notion of
directional derivatives in the continuous setting. We define $f_\ell^{(r)}(\ba,\bx)$ to simplify the statement of
Lemma~\ref{lem:taylor}.
From this definition, and by using the fact that the degree of $f$ is smaller than $\char(\F)$,
it is easy to see the following observations:

\begin{obs}\label{obs:der via hom}
$f_i^{(1)}(\ba,\bx)=H^{(i-1)}(f_i(\bx+\ba))$.
Thus, $f^{(1)}(\ba,\bx)=\sum_{i=1}^{\deg(f)}H^{(i-1)}(f_i(\bx+\ba))$.
\end{obs}

% This recursion does not work on fields with small characteristic (smaller than degree of f)
\begin{obs}\label{obs:derivative}
The polynomials $f^{(r)}(\ba,\bx)$ have the following recursive structure:
\begin{equation}\label{eq:der}
f^{(r+1)}(\ba,\bx) \equiv \sum_{j=1}^n a_j\cdot \frac{\D (f^{(r)}(\ba, \bx))}{\D x_j}.
\end{equation}
\end{obs}

This recursive structure implies that the directional derivatives of lower order exhibit a ``domino effect," which
can be captured in the following observation:

\begin{obs}\label{obs:difshifts}
If $f^{(1)}(\ba, \bx) \equiv f^{(1)}(\bb, \bx)$ then $f^{(r)}(\ba, \bx) \equiv f^{(r)}(\bb, \bx)$, for all $r \ge 1$.
\end{obs}

\begin{proof} We will prove this observation by induction on $r$. We know that the claim is true for $r=1$. Now,
given that the claim is true for all values $1 \le t \le r$, we have:
\begin{align*}
f^{(r+1)}(\ba, \bx) &\equiv \sum_{j=1}^n a_j \cdot \frac{\D (f^{(r)}(\ba, \bx))}{\D x_j} &
(\text{by observation~\ref{obs:derivative}}) \\
&\equiv  \sum_{j=1}^n a_j \cdot \frac{\D (f^{(r)}(\bb, \bx))}{\D x_j} & (\text{by induction hypothesis on } r) \\
&\equiv \sum_{j=1}^n a_j \cdot \frac{\D}{\D x_j}\left( \sum_{\be \in [n]^r}
\left( \prod_{k=1}^r b_{e_k} \right) \cdot \frac{\D^{r} f}{\D x_{e_1} \ldots \D x_{e_r}}(\bx)  \right) &
(\text{by definition~\ref{def:dirder}}) \\
&\equiv \sum_{\be \in [n]^r}
\left( \prod_{k=1}^r b_{e_k} \right) \cdot \frac{\D^{r} }{\D x_{e_1} \ldots \D x_{e_r}} \left(
\sum_{j=1}^n a_j \cdot \frac{\D f}{\D x_j}(\bx) \right) &
(\text{by rearranging the sum}) \\
&\equiv \sum_{\be \in [n]^r}
\left( \prod_{k=1}^r b_{e_k} \right) \cdot \frac{\D^{r} }{\D x_{e_1} \ldots \D x_{e_r}} \left( f^{(1)}(\ba,\bx) \right) &
(\text{by definition~\ref{def:dirder}}) \\
&\equiv \sum_{\be \in [n]^r}
\left( \prod_{k=1}^r b_{e_k} \right) \cdot \frac{\D^{r} }{\D x_{e_1} \ldots \D x_{e_r}} \left( f^{(1)}(\bb,\bx) \right) &
(\text{by induction hypothesis}) \\
&\equiv \sum_{\be \in [n]^r}
\left( \prod_{k=1}^r b_{e_k} \right) \cdot \frac{\D^{r} }{\D x_{e_1} \ldots \D x_{e_r}} \left(
\sum_{j=1}^n b_j \cdot \frac{\D f}{\D x_j}(\bx) \right) &
(\text{by definition~\ref{def:dirder}}) \\
&\equiv \sum_{\be \in [n]^{r+1}} \left( \prod_{k=1}^{r+1} b_{e_k} \right) \cdot
\frac{\D^{r+1} f}{\D x_{e_1} \ldots \D x_{e_{r+1}}}(\bx) & (\text{by rearranging the sum}) \\
&\equiv f^{(r+1)}(\bb, \bx) & (\text{by definition~\ref{def:dirder}})
\end{align*}
and this concludes the inductive proof.
\end{proof}

Observation~\ref{obs:difshifts} tells us that if the first order directional derivatives are equal for two different
directions $\ba$ and $\bb$, then all of the higher-order directional derivatives will also be equal. This observation
will be crucial in the design of our algorithm.

Now that we defined directional derivatives, we can state the main lemma of this section, which gives us relations
between the homogeneous components of $f(\bx)$ and $f(\bx+\ba)$:
\newpage
\begin{lem}[Taylor Expansion Lemma]\label{lem:taylor}
Let $\ba \in \F^n$ and let $f \in \F[\bx]$ be such that $deg(f) = d < \char(\F)$
Then, the following relations hold for all $0 \le i \le d$:
\begin{equation}\label{eq:taylor2}
H^{i}(f(\bx + \ba)) \equiv \sum_{j=i}^d \frac{1}{(j-i)!} \cdot  f_j^{(j-i)}(\ba, \bx)
\end{equation}
\end{lem}

\begin{proof}
	Notice that it is enough to show this lemma for the case where $f$ is a single monomial, since the general
	case follows by additivity of partial derivatives.
	
	If $f(\bx) = \prod_{j=1}^n x_j^{d_j}$, we
	have that
	\[ f(\bx+\ba) = \prod_{k=1}^n (x_k+a_k)^{d_k}  \]
	which implies
	\begin{align*}
	H^i(f(\bx+\ba)) &= \sum_{\substack{j_1+j_2+\ldots+j_n = i \\ j_k \ge 0, \ k \in [n]}}\;\;
	\prod_{k=1}^n \binom{d_k}{j_k} a_k^{d_k-j_k} x_k^{j_k} \\
	&= \sum_{\substack{j_1+j_2+\ldots+j_n = i \\ j_k \ge 0, \ k \in [n]}}
	\left(\prod_{k=1}^n \frac{1}{(d_k-j_k)!} a_k^{d_k-j_k}\right) \cdot
	\left(\prod_{k=1}^n \frac{d_k!}{j_k!} x_k^{j_k}\right) \\
	&= \sum_{\substack{j_1+j_2+\ldots+j_n = i \\ j_k \ge 0, \ k \in [n]}}
	\left(\prod_{k=1}^n \frac{1}{(d_k-j_k)!} a_k^{d_k-j_k}\right) \cdot \frac{\D^{d-i} f}{\prod_{k \in [n]}(\D x_k)^{d_k-j_k}}(\bx) \\
	&=^{(*)} \frac{1}{(d-i)!} \cdot \sum_{\be \in [n]^{d-i}}
	\left(\prod_{k=1}^{d-i} a_{e_k} \right) \cdot \frac{\D^{d-i} f}{\D x_{e_1} \ldots \D x_{e_{d-i}}}(\bx) \\
	&= \frac{1}{(d-i)!} \cdot f^{(d-i)}(\ba,\bx),
	\end{align*}
	where equality $(*)$ follows as each term $\prod_{k=1}^n a_k^{d_k-j_k}$ is counted
	${d-i \choose d_1-j_1,\ldots,d_n-j_n}$ many times when in the sum $\sum_{\be \in [n]^{d-i}}
	\left(\prod_{k=1}^{d-i} a_{e_k} \right)$.
	
	The equations above and additivity of partial derivatives imply that the lemma is true when $f$ is a homogeneous
	polynomial. The proof of the general case is as follows:
	\begin{align*}
	H^i(f(\bx+\ba)) &= H^i\left(\sum_{j=i}^{d} H^j(f)(\bx+\ba)\right) = \sum_{j=i}^{d} H^i\left(H^j(f)(\bx+\ba)\right) \\
	&= \sum_{j=i}^d \frac{1}{(j-i)!} \cdot f_{j}^{(j-i)}(\ba,\bx)
	\end{align*}
	Where the last equality is true because the lemma is true for homogeneous polynomials, and each
	$H^j(f)$ is a homogeneous polynomial of degree $j$.
\end{proof}

Lemma~\ref{lem:taylor} can be seen as the multivariate Taylor expansion of the polynomial $f(\bx)$ around the
point $\ba \in \F^n$.

In order to use Lemma~\ref{lem:taylor}, we need to have access to the polynomials $f_k^{(r)}(\ba, \bx)$ defined in
the lemma. The following observations allow us to have the type of accesses we need.

\begin{obs}\label{obs:bb-hom}
The polynomials $f_k^{(r)}(\ba,\bx)$ are a constant multiple of the homogeneous components of degree 
$k-r$ of $H^k(f)(\bx+\ba)$.
\end{obs}

This observation is important because given black-box access to $f$, we can obtain black-box access to the polynomials
$f_k^{(r)}(\ba,\bx)$ by interpolation of the polynomials $H^k(f)(\bx+\ba)$, as we do in Lemma~\ref{lem:interp}.

Notice that if we are only concerned with a bound on the size of a circuit computing the homogeneous components of a
polynomial $f$, then by a result of Strassen in \cite{Strassen73a} we have the following theorem:

\begin{thm}\label{thm:strassen}
If $f$ can be computed by an arithmetic circuit of size $s$, then for every $k \in \N$, there is a homogeneous circuit
of size at most $O(r^2 s)$ computing all of the polynomials $H^k(f)$, where $0 \le k \le r$. Moreover, given access to
the circuit computing $f$, we can construct the homogeneous circuit computing the homogeneous components of $f$.
\end{thm}

A straightforward consequence of this theorem and of observation~\ref{obs:bb-hom} is stated below:

\begin{cor}\label{cor:circuit-hom}
If $f$ has degree $d$ and can be computed by an arithmetic circuit of size $s$, then
for every shift $\ba \in \F^n$, there is a homogeneous circuit of size at most $O(d^4 s)$ computing all of the
polynomials $f^{(r)}_k(\ba, \bx)$, where $0 \le k, r \le d$. Moreover, given access to
the circuit computing $f$, we can construct the homogeneous circuit computing all the polynomials $f^{(r)}_k(\ba, \bx)$.
\end{cor}

%%%%%%%%%%%%%%%%%%%%%%%%%%%%%%%%%%%%%%%%%%%%%%%
%%%%%%%%%%%%%%%%%%%%%%%%%%%%%%%%%%%%%%%%%%%%%%%
%%%%%%%%%%%%%%%%%%%%%%%%%%%%%%%%%%%%%%%%%%%%%%%
\section{Kernel of Shifts of a Polynomial}\label{sec:kershifts}

As observed by Grigoriev in \cite{grig-q2tcs}, the set of points $\ba \in \F^n$ such that $f(\bx+\ba) \equiv f(\bx)$, which
we call the \emph{kernel} of $f$, forms a linear subspace of $\F^n$.
In this section, we describe some properties of the kernel and introduce
some lemmas which describe the relationship between points $\ba$
in the kernel and the directional derivatives of $f$ on the direction $\ba$.

We begin with the following definitions:

\begin{define}\label{def:ker}
Let $f(\bx) \in \F[\bx]$ be a polynomial of degree $d$. We define the \emph{kernel} of $f$ as the set
\[ \cS_f = \{ \ba \in \F^n \ | \ f(\bx + \ba) \equiv f(\bx) \}, \]
that is, $\cS_f$ is the set of all points $\ba \in \F^n$ such that if we shift the input of $f(\bx)$ by $\ba$
we obtain the same formal polynomial. Here $\bx$ is regarded as a formal set of variables and $\ba$
is a point in $\F^n$.
\end{define}

We can observe the following properties of the kernel $\cS_f$:

\begin{obs}
Let $f \in \F[\bx]$ be a polynomial of degree $d$, where $d < \char(\F)$ if $\char(\F)\neq 0$.
Then the kernel $\cS_f$ is a subspace of $\F^n$.
\end{obs}
\begin{proof} We need to check three conditions:
\begin{enumerate}[(i)]
	\item $0 \in \cS_f$
	\item $\ba, \bb \in \cS_f \then \ba+\bb \in \cS_f$
	\item $\ba \in \cS_f \then t \ba \in \cS_f$, for all $t \in \F$
\end{enumerate}
Conditions (i) and (ii) are trivial to check. Hence, we only need to show that condition (iii) holds. By repeatedly applying (ii),
we have $\ba \in \cS_f \then k \ba \in \cS_f$, for any $k \in \N$. In particular, since the degree of $f$ is less then the
characteristic of $\F$, we have that $k \ba \in \cS_f$ for $0 \le k \le d$, which are all distinct values.
Hence, the polynomial $f(\bx + t\ba) - f(\bx) \in \F(\bx)[t]$ has degree $\le d$ in $t$ and has at least $d+1$ distinct roots. This
implies that $f(\bx + t\ba) - f(\bx)$ must vanish as a polynomial in $t$, which implies that $f(\bx + t\ba) - f(\bx) \equiv 0$
for all $t \in \F$. This proves condition (iii) and therefore $\cS_f$ is a subspace of $\F^n$.
\end{proof}

More generally, we can define the set of shift-equivalences between two polynomials $f, g \in \F[\bx]$:
\begin{define}
Given two polynomials $f, g \in \F[\bx]$, we define the set of shift-equivalences of $f$ and $g$ as
\[ \cS_{f,g} = \{ \ba \in \F^n \ | \ f(\bx + \ba ) \equiv g(\bx) \}, \]
that is, $\cS_{f,g}$ is the set of all points $\ba \in \F^n$ such that if we shift the input of
$f(\bx)$ by $\ba$ we obtain the formal polynomial $g(\bx)$.
\end{define}

Note that the set $\cS_{f,g}$ is intrinsically related to $\cS_f$, since if we have any two elements $\ba$ and $\bb$
of $\cS_{f,g}$, we must have that $\ba-\bb \in \cS_f$. In other words, $\cS_{f,g}$ is a coset of $\cS_f$. Furthermore,
it must hold that $\cS_f=\cS_g$.

\begin{lem}\label{lem:Sp=Sq}
Let $f, g \in \F[\bx]$ such that there exists $\ba\in \F^n$ for which $f(\bx+\ba)=g(\bx)$. Then $\cS_f=\cS_g$.
Furthermore, for any such $\ba$ we have $\cS_{f,g} =\cS_f+\ba$.

\end{lem}

\begin{proof}
Let $\bb\in \cS_g$. Then $f(\bx+\bb) = f((\bx-\ba+\bb)+\ba) = g(\bx-\ba+\bb)=g(\bx-\ba)=f(\bx)$. Hence,
$\cS_g\subseteq \cS_f$. The other direction is similar.

Given $\ba$ and $\bb$ in $\cS_{f,g}$ we have $f(\bx-\ba+\bb) = g(\bx-\ba) = f(\bx)$. Hence $\bb-\ba \in \cS_{f}$. Thus,
$\cS_{f,g} \subseteq \cS_f+\ba$. It is also straightforward to verify that $\cS_f+\ba \subseteq \cS_{f,g}$.
\end{proof}

Another interesting property which relates the kernel to directional derivatives is captured by the following lemma,
which states that a shift $\ba \in {\F}^n$ is in the kernel of shifts $\cS_f$ if, and only if, the first order directional
derivative of $f$ in the direction $\ba$ is zero.

\begin{lem}\label{lem:kernelshift} Let $f \in \F[\bx]$ be a polynomial of degree $d$, where $d < \char(\F)$ if $\char(\F)\neq 0$.
Then, $\ba \in \cS_f$ if, and only if, $f^{(1)}(\ba, \bx) \equiv 0$.
\end{lem}

\begin{proof}
After a suitable change of basis that maps $\ba$ to $\be_1$, we need to prove that
$\be_1 \in \cS_f$ if, and only if, $f^{(1)}(\be_1, \by) \equiv 0$, where $\by$ is the image of $\bx$ under this change
of basis. Since $f^{(1)}(\be_1, \by) = \dfrac{\D f}{\D y_1}(\by)$, we must show that $\be_1 \in \cS_f$ if, and only if, 
$\dfrac{\D f}{\D y_1}(\by) \equiv 0$. 

To see the first direction, note that if $\dfrac{\D f}{\D y_1}(\by) \equiv 0$ and
$d < \char(\F)$ then $f \in \F[y_2, \ldots, y_n]$ which implies $f(\by+\be_1) \equiv f(\by)$. Hence, $\be_1 \in \cS_f$.

On the other hand, let us write $f(\by) = \ds\sum_{i=0}^d y_1^i \cdot f_i(y_2, \ldots, y_n)$.
Let $k$ be the highest index for
which $f_k(y_2,\ldots, y_n) \neq 0$. If $\dfrac{\D f}{\D y_1}(\by) \not\equiv 0$, then $k>0$.
Let $\bc \in \F^n$ be such that $f_k(c_2, \ldots, c_n) \neq 0$ and $c_1 = 0$.
Then, if we define $b_i = f_i(c_2, \ldots, c_n)$, for $0 \le i \le k$, we have
$f(\bc) = b_0$ and $f(\bc+t\be_1) = b_k t^k + \ds\sum_{i=1}^{k-1} b_i t^i $, where $b_k \neq 0$. This implies
$f(\bc+t\be_1) - f(\bc) = b_k t^k + \ds\sum_{i=1}^{k-1} b_i t^i$. Hence, there exists $t \in \F$ such that
$b_k t^k + \ds\sum_{i=1}^{k-1} b_i t^i \neq 0 \then f(\bc+t\be_1) - f(\bc) \neq 0$. Thus,
$f(\by+t\be_1) - f(\by) \not\equiv 0$ and so $\be_1 \not\in \cS_f$.
\end{proof}
An easy corollary of Lemma~\ref{lem:kernelshift} and of observation~\ref{obs:difshifts} is the following:

\begin{cor}\label{cor:hokernelshift} If $\ba \in \cS_f$ then $f^{(r)}(\ba, \bx) \equiv 0$, for all $r \ge 1$.
\end{cor}

\ignore{
Notice that, unlike the kernel $\cS_f$, if $f$ and $g$ are distinct polynomials, then $\cS_{f,g}$ will not
necessarily be a subspace of $\F^n$. However, $\cS_{f,g}$ has a nice property which is captured by the following
lemma, which is trivial to prove.

\begin{lem}\label{lem:shiftlemma}
If $\ba \in \cS_{f,g}$ and $\bb \in \F^n$ then $\bb \in \cS_{f,g}$ if, and only if, $\ba-\bb \in \cS_{f} \cap \cS_{q}$.
\end{lem}
}

Another property that easily follows from linearity of $f^{(1)}(\ba, \bx)$ (in $\ba$) and from Lemma~\ref{lem:Sp=Sq}
is captured by the following lemma:
\begin{lem}\label{lem:mainshiftlemma}
If $\ba \in \cS_{f,g}$ and $\bb \in \F^n$ then $\bb \in \cS_{f,g}$ if, and only if,
$f^{(1)}(\ba, \bx) \equiv f^{(1)}(\bb, \bx)$. Thus, $\bb\in\cS_{f,g}$ if, and only if,
$f_i^{(1)}(\bb, \bx) \equiv f_i^{(1)}(\ba, \bx)$ for all $0 \le i \le d$.
\end{lem}

\begin{proof}
$f^{(1)}(\ba, \bx) \equiv f^{(1)}(\bb, \bx)$ if, and only if, $f^{(1)}(\bb-\ba, \bx) \equiv 0$. By Lemma~\ref{lem:kernelshift}
this is equivalent to $\bb-\ba \in \cS_f$ and hence to $\bb \in \cS_f + \ba$. This is equivalent, by Lemma~\ref{lem:Sp=Sq},
to having $\bb \in \cS_{f,g}$ as desired.

The second part of the lemma is immediate.
\end{proof}

%%%%%%%%%%%%%%%%%%%%%%%%%%%%%%%%%%%%%%%%%%%%%%%%%%%%
%%%%%%%%%%%%%%%%%%%%%%%%%%%%%%%%%%%%%%%%%%%%%%%%%%%%
%%%%%%%%%%%%%%%%%%%%%%%%%%

\section{Proof of Equivalence Under Shifts}\label{sec:main}

In this section we give intuition and an overview of our algorithm in subsection~\ref{subsec:overview}, followed
by a formal description of the algorithm and its analysis in subsection~\ref{subsec:mainproof}.
\newpage
%%%%%%%%%%%%%%%%%%%%%%%%%%%%%%%%%%%%%%%%%%%%%%%%
\subsection{Overview of the Algorithm}\label{subsec:overview}

In this section, we will describe an overview of the steps in our algorithm.
The high level idea of the algorithm was given in section~\ref{sec:proof overview}.
For the sake of clarity, we will leave the explanations of the preprocessing stage for the analysis of the algorithm and
we will assume that the input given is already preprocessed accordingly.

In the highest level, our algorithm will produce a candidate shift $\ba$ such that $f(\bx+\ba) \equiv g(\bx)$
and then use PIT on the polynomial $f(\bx) - g(\bx-\ba)$, to check that the solution $\ba$ is indeed a good shift.
We need to perform the PIT on the polynomial $f(\bx) - g(\bx-\ba)$ because $\cM_2$ is closed under shifts.
We proceed in this way because
this approach allows us to assume from the beginning on that $\cS_{f,g} \neq \emptyset$. For this section, we
can assume that $\cS_{f,g} \neq \emptyset$, that $\bc \in \cS_{f,g}$ and that $d_f = d_g = d$.

By our Taylor Expansion Lemma (Lemma~\ref{lem:taylor}), to find a good shift $\bc \in \cS_{f,g}$
we need to solve the system
of polynomial equations (in the variables $\ba$) given by the set of equations~\eqref{eq:taylor2} in the Lemma.
We cannot hope to solve these equations directly, since that would involve solving non-linear systems of equations.
However, Lemma~\ref{lem:mainshiftlemma} tells us that in order to find a good
shift, we only need to obtain black-box access to the polynomials $f_k^{(1)}(\bc,\bx)$, where $\bc \in \cS_{f,g}$.
If we succeed in obtaining black-box access to these polynomials, finding a good shift will only involve solving a
linear system of polynomial equations in the black-box setting, which we can do by any of the lemmas: 
Lemma~\ref{lem:poleq-wb}, Lemma~\ref{lem:poleq-bb} or Lemma~\ref{lem:poleq-rand}, depending on which case we 
are in. Hence, our approach to solve the original set of equations is to obtain black-box
access to the polynomials $f_k^{(1)}(\bc,\bx)$.

Note that we cannot obtain direct access to $f_k^{(1)}(\bc,\bx)$ (through interpolation) from neither $f$ nor $g$, for
a general $k$. However, from $f$ and $g$ we have black-box access to $f_d^{(1)}(\bc,\bx)$, since
$f_d^{(1)}(\bc,\bx) = H^{d-1}(g(\bx)) - H^{d-1}(f(\bx))$. It turns out that this initial information is enough for the algorithm
to find a good shift. To find the shift we will iteratively find candidate solutions $\ba_r$, such that
$f_k^{(1)}(\bc,\bx) \equiv f_k^{(1)}(\ba_r,\bx)$ for all $d-r \le k \le d$. Then, by the domino effect from
Observation~\ref{obs:difshifts} we have that $f_k^{(t)}(\bc,\bx) \equiv f_k^{(t)}(\ba_r,\bx)$, for all $t \ge 1$. Hence,
once we find $\ba_k$ the domino effect and Lemma~\ref{lem:taylor} imply that we can find
$\ba_{r+1}$ simply by solving linear equations. In the end, if the algorithm does not fail, we will obtain $\ba_d$ such
that $f_k^{(1)}(\bc,\bx) \equiv f_k^{(1)}(\ba_d,\bx)$ for all $0 \le k \le d$, and thus by Lemma~\ref{lem:mainshiftlemma}
we must have $\ba_d \in \cS_{f,g}$. This domino effect lies at the crux of the proof of correctness of our algorithm.

%%%%%%%%%%%%%%%%%%%%%%%%%%%%%%%%%%%%%%%%%%%%%%%%
\subsection{Formal Description and Proof of Correctness}\label{subsec:mainproof}

For simplicity, we will describe the algorithm receiving the input already preprocessed. 
%The algorithm (Algorithm~\ref{alg:main}) is given in page \pageref{alg:main}.  We first discuss the preprocessing step. 

%\vspace{10pt}

\begin{algorithm}[H]
	\caption{Main Algorithm} \label{alg:main}
 	\SetAlgoLined
	\KwIn{black-boxes (or white-boxes) for polynomials $f \in \cM_1 , g \in \cM_2$, and degree
	of $f$, which we denote by $d$.}
	\KwOut{a non-zero shift in $\cS_{f,g}$, if one exists, or FAIL, if $\cS_{f,g} = \emptyset$.}
	By interpolation, obtain black-box access to the homogeneous components $H^0(f), H^1(f), \ldots, H^d(f)$ and
	$H^0(g), H^1(g), \ldots, H^d(g)$. \\
	Set $\ba_0 \gets 0$. \\
 	\For{k = $1, \ldots, d$}{
		Solve, via any appropriate lemma from subsections~\ref{sec:poleq-wb} or~\ref{sec:poleq-bb}\footnote{If we
		are in the white-box setting, we will use Lemma~\ref{lem:poleq-wb}, if we are in the black-box setting and 
		we have a hitting set, then we will use Lemma~\ref{lem:poleq-bb}, or if we are in the black-box setting and are
		allowed randomness we will use Lemma~\ref{lem:poleq-rand}.}, 
		the linear system given by the following equations,
		where in these equations the variables are the entries of $\bb$ and we have one equation
		for each $i$ such that $d-k \le i \le d$.
		\begin{align}\label{eq:algmain}
		H^i(g(\bx)) &= H^i(f(\bx)) + f_{i+1}^{(1)}(\bb,\bx) + \sum_{j=i+2}^{d} \frac{1}{(j-i)!}
		\cdot f_j^{(j-i)}(\ba_{k-1}, \bx)
		\end{align}
		If system of equations~\eqref{eq:algmain} has no solution, \Return FAIL \\
		Else, $\ba_k \gets \bb$.
 	}
	Perform PIT on $f(\bx) - g(\bx- \ba_d)$. \\
	If $f(\bx) - g(\bx-\ba_d) \equiv 0$, \Return $\ba_d$ \\
	Else, \Return FAIL
\end{algorithm}

%\vspace{15pt}

\paragraph{Preprocessing Stage:}

In case $\cS_{f,g} = \emptyset$, our algorithm might do meaningless computations,
but because we will verify our candidate solution, our algorithm is sure to return that no shift
exists in the end. Notice that if we have $d_f \neq d_g$, even the interpolation step that we perform in the beginning
will err when computing homogeneous components of $g$ (because we will not interpolate with the proper degree). However,
this is ok because we have the PIT step in the end, which will prevent us from returning any wrong answers that
may arise from the meaningless computations.

Hence, from now on we can, and will, assume that $\cS_{f,g} \neq \emptyset$. In particular, this implies that
we can assume that there exists $\bc \in \cS_{f,g}$ and that $d_f = d_g = d$. Thus, our algorithm will assume that
the input is given by two polynomials $f,g \in \F[\bx]$ of degree upper bounded by the degree of $f$, which we
will denote by $d$. Notice that we can also assume that $d$ is the exact degree of $f$, since from the upper bound
on the degree we can interpolate $f$ and perform PIT on each homogeneous components of $f$ (recall $\cM_1$ is
closed under homogeneous components). Then, the degree of $f$ will be the value of the highest non vanishing
homogeneous component. Thus, we will assume that $d$ is the exact degree of $f$, as opposed to an upper bound.

\paragraph{Analysis of the Algorithm in the Black-Box case, with a Hitting Set:}

\begin{proof}
Notice that if $\cS_{f,g} = \emptyset$, then even if the algorithm finishes the \textbf{for} loop,
it will return FAIL,  since PIT on $f(\bx) - g(\bx-\ba_d)$ will return that the polynomial is non-zero. Therefore,
we never err in this case.
Hence, for the  rest of the analysis, let us assume that $\cS_{f,g} \neq \emptyset$. This implies that there
exists a shift $\bc \in \cS_{f,g}$.
Since $g(\bx)\equiv f(\bx+\bc)$ it holds that $ H^i(g(\bx)) \equiv  H^i(f(\bx+\bc))$.
From Lemma~\ref{lem:taylor}, we have
\[ H^i(f(\bx+\bc)) \equiv \sum_{j=i}^d \frac{1}{(j-i)!} \cdot  f_j^{(j-i)}(\bc, \bx)  \equiv H^i(f(\bx)) +
f_{i+1}^{(1)}(\bc,\bx) + \sum_{j=i+2}^d \frac{1}{(j-i)!} \cdot  f_j^{(j-i)}(\bc, \bx)  .\]
Hence,
\begin{equation}\label{eq:simp}
 H^i(g(\bx)) - H^i(f(\bx)) - \sum_{j=i+2}^d \frac{1}{(j-i)!} \cdot  f_j^{(j-i)}(\bc, \bx) \equiv
f_{i+1}^{(1)}(\bc,\bx)
\end{equation}
for all $0 \le i \le d$.

Recall, that by Lemma~\ref{lem:mainshiftlemma}, to find a shift in $\cS_{f,g}$ it is enough to
find a shift $\bb$ such that $f_i^{(1)}(\bb, \bx) \equiv f_i^{(1)}(\bc, \bx)$ for all $0 \le i \le d$.

Observation~\ref{obs:difshifts} implies that if $f_i^{(1)}(\bb, \bx) \equiv f_i^{(1)}(\bc, \bx)$
for some $0 \le i \le d$, then $f_i^{(r)}(\bb, \bx) \equiv f_i^{(r)}(\bc, \bx)$ for all $r \ge 1$, for this particular $i$.
Therefore, if we show that our algorithm maintains the invariant
\begin{equation}\label{eq:inv}
f_i^{(1)}(\ba_k, \bx) \equiv f_i^{(1)}(\bc, \bx), \ \ \   \text{for all } i \st d-k+1 \le i\le d
\end{equation}
at every iteration of the loop,  then at the end of the loop
we will have $f_i^{(1)}(\ba_d, \bx) \equiv f_i^{(1)}(\bc, \bx)$ for all $0 \le i \le d$, which
by Lemma~\ref{lem:mainshiftlemma} implies that $\ba_d \in \cS_{f,g}$. Thus, it is enough to show that
our algorithm preserves invariant \eqref{eq:inv}.

For $k=1$, we need to solve equations $H^{d-1}(g(\bx)) \equiv H^{d-1}(f(\bx)) + f_{d}^{(1)}(\bb,\bx)$
and $H^{d}(g(\bx)) \equiv H^{d}(f(\bx))$. Notice that equation $H^{d}(g(\bx)) \equiv H^{d}(f(\bx))$ is always
true, due to the assumptions we are making about our input after preprocessing. Therefore, we will not mention
this equation anymore and the only relevant polynomial equation to solve in this case is
$H^{d-1}(g(\bx)) \equiv H^{d-1}(f(\bx)) + f_{d}^{(1)}(\bb,\bx)$.

By identity~\eqref{eq:simp} we have that $H^{d-1}(g(\bx)) - H^{d-1}(f(\bx)) \equiv f_d^{(1)}(\bc,\bx)$, which is
a directional derivative of $f$ and therefore is an element of $\cM_1$. Notice that, for each $\ba \in \F^n$,
\[  \ds\sum_{j=1}^n a_j \cdot \frac{\D H^d(f(\bx))}{\D x_j} \equiv f_d^{(1)}(\ba, \bx) \in \cM_1 \]
which implies that we have PIT for linear combinations of the partial derivatives of $H^d(f(\bx))$.
Thus, by solving the polynomial equation $H^{d-1}(g(\bx)) - H^{d-1}(f(\bx)) \equiv f_d^{(1)}(\bb,\bx)$ on the variables
$\bb$, using Lemma~\ref{lem:poleq-bb}, we get a solution $\ba_1$ such that
$H^{d-1}(g(\bx)) \equiv H^{d-1}(f(\bx)) + f_{d}^{(1)}(\ba_1,\bx)$ (since we know $\cS_{f,g} \neq \emptyset$). Notice that
we can use Lemma~\ref{lem:poleq-bb} because we have black-box access to all polynomials in the equation.
Hence, we have a solution $\ba_1$ such that
\[ f_{d}^{(1)}(\ba_1,\bx) \equiv H^{d-1}(g(\bx)) - H^{d-1}(f(\bx))  \equiv  f_{d}^{(1)}(\bc,\bx) \]
and hence, our invariant~\eqref{eq:inv} holds true in the first case.

Now, assume that our invariant is true for $\ba_{k-1}$, $k \ge 2$.
At the $k^{th}$ iteration, equations~\eqref{eq:algmain} are equivalent to
\begin{align*}
	H^i(g(\bx)) &= H^i(f(\bx)) + f_{i+1}^{(1)}(\bb,\bx) + \sum_{j=i+2}^{d} \frac{1}{(j-i)!}
		\cdot f_j^{(j-i)}(\ba_{k-1}, \bx) \\
		&\equiv H^i(f(\bx)) + f_{i+1}^{(1)}(\bb,\bx) + \sum_{j=i+2}^{d} \frac{1}{(j-i)!}
		\cdot f_j^{(j-i)}(\bc, \bx) & \forall d-k \le i \le d.
\end{align*}

Where in the last equality we used the fact that our invariant holds for $\ba_{k-1}$ together with
Observation~\ref{obs:difshifts}. These equations, together with equation~\eqref{eq:simp} imply:

\begin{align*}
	f_{i+1}^{(1)}(\bb,\bx) &\equiv H^i(g(\bx)) - H^i(f(\bx)) - \sum_{j=i+2}^{d} \frac{1}{(j-i)!}
		\cdot f_j^{(j-i)}(\ba_{k-1}, \bx) \\
		&\equiv f_{i+1}^{(1)}(\bc,\bx) , \ \text{ for all } d-k \le i \le d
\end{align*}
In other words, for all $d-k \le i \le d$
$$\sum_{\ell=1}^n b_\ell \cdot \frac{\D f_{i+1}}{\D x_\ell}(\bx) = f_{i+1}^{(1)}(\bb,\bx) \equiv f_{i+1}^{(1)}(\bc,\bx)$$

\ignore{Notice that  $ f_{i+1}^{(1)}(\bb,\bx)$ belongs to $\cM_1$ as a directional derivative of some homogeneous component of
a polynomial computed in $\cM_1$. Clearly, $g^{(1)}_{i+1}(\bx)\in \cM_2$, and therefore using the hitting set for
polynomials in $\cM_1+\cM_2$ we can use  Lemma~\ref{lem:poleq-bb} to solve this system of equations in $\bb$.
Thus, we can solve the system of polynomial equations~\eqref{eq:algmain}.
Hence, we have found a vector $\bb$ such that
$f_{i}^{(1)}(\bb,\bx) \equiv f_{i}^{(1)}(\bc,\bx), \ \text{ for all } d-k+1 \le i \le d$. Since we set $\ba_k \gets \bb$, we
have that our algorithm maintains invariant~\eqref{eq:inv} at iteration $k$.}

Notice that both sides of each of the equations above belong to the circuit class $\cM_1$,
as both sides are first-order directional derivatives of $H^{i+1}(f(\bx))$.
Since we have black-box access to both sides of the
equations above, Lemma~\ref{lem:poleq-bb} and PIT for $\cM_1$ imply that we can solve
the system of polynomial equations~\eqref{eq:algmain}.

Thus, since the invariant is maintained until the end, we must have that
$\ba_d$ is such that $f_{i}^{(1)}(\ba_d,\bx) \equiv f_{i}^{(1)}(\bc,\bx), \text{ for all } 0 \le i \le d$,
for some $\bc \in \cS_{f,g}$. By Lemma~\ref{lem:mainshiftlemma} we must have that $\ba_d \in \cS_{f,g}$.

\paragraph{Runtime Analysis:}

Notice that we iterate through the loop $d$ times and at each iteration we solve a linear system of at
most $d \cdot |\cH_1|$ equations in $n$ variables, where $\cH_1 \subset {\F}$ is a hitting set for the
circuit class $\cM_1$. After exiting the loop, we only need to perform PIT on $f(\bx) - g(\bx-\ba_d)$, which we
assumed it takes polynomial time, for we have PIT for polynomials of the form $f-g$, where $f \in \cM_1$
and $g \in \cM_2$. Hence, the total running time is polynomial in the size of the input.

\end{proof}

%%%%%%%%%%%%%%%%%%%%%%%%%%%%%%%%%%%%%%%%%%%%%%%%
\paragraph{Analysis in the Randomized Case:} The randomized case is analogous to the deterministic black-box case. 
Whenever we need to perform PIT in our main algorithm, we will use Lemma~\ref{lem:sz}. Whenever we need to solve 
a system of polynomial equations given black box access to the polynomials in question, we will use 
Lemma~\ref{lem:poleq-rand} (when we are allowed randomness), instead of Lemma~\ref{lem:poleq-bb} 
(which handles the case when we have a hitting set). 

We need to solve $d$ systems of polynomial equations and we perform the PIT algorithm as in 
Lemma~\ref{lem:sz} $O(d)$ times. Hence, by setting the error parameter each time as $\veps/d^2$ and by a 
union bound, our algorithm will err with probability at most $\veps$. Since the amount of randomness that
we need to solve a polynomial system or to perform PIT is polynomial in the logarithm of the error parameter, this gives us 
the desired running time as claimed in Theorem~\ref{thm:main-rand-int}.

%%%%%%%%%%%%%%%%%%%%%%%%%%%%%%%%%%%%%%%%%%%%%%%%
\paragraph{Analysis in the White-Box Case:} Notice that by Theorem~\ref{thm:strassen} and
Corollary~\ref{cor:circuit-hom}, given access to circuits computing $f, g$ implies that we also have access to
circuits computing the polynomials $f_\ell^{(r)}(\ba, \bx)$ and $H^\ell(g)$.
Thus, we also have white-box access to linear combinations of $m = \max(n, d)$ of these polynomials.

After the step above, the white-box case is analogous to the deterministic case. Whenever we need to perform
PIT in our main algorithm, we will use the appropriate PIT algorithm for the white box class that we are considering. For instance,
whenever the algorithm above uses PIT for the class $\cM_1$, we will use the white-box algorithm, and the same happens
for the other classes. In addition, whenever we need to solve a linear system of polynomial equations, we will use the method in
section~\ref{sec:poleq-wb} to solve our system.
Thus, the same argument as the one given above for the deterministic black-box case will go through, even for the
preprocessing stage, and therefore we are done.

%%%%%%%%%%%%%%%%%%%%%%%%%%%%%%%%%%%%%%%%%%%%%%%%%%%%
%%%%%%%%%%%%%%%%%%%%%%%%%%%%%%%%%%%%%%%%%%%%%%%%%%%%
%%%%%%%%%%%%%%%%%%%%%%%%%%

\section{Conclusion and Open Questions}\label{sec:conclusion}

In this paper, we reduced the problem of shift-equivalence to the problem of solving PIT, and as a consequence of this
reduction we obtained a polynomial-time randomized algorithm for the shift-equivalence problem, over characteristic zero or when the
characteristic of the base field is larger than the degrees of the polynomials.

We gave some examples for classes of circuits where this can be performed deterministically in quasi-polynomial time. One example where we ``almost'' have such a result is when testing whether a given sparse polynomial is equivalent to a shift of another sparse polynomial. Note that while the class of sparse polynomials is closed under partial derivatives and homogeneous components, it is not closed under shifts and so we cannot use our approach. Nevertheless, it is quite likely that this simple case can be solved using other techniques. 

%%%%%%%%%%%%%%%%%%%%%%%%%%%%%%%%%%%%%%%%%%%%%%%%%%%%
%%%%%%%%%%%%%%%%%%%%%%%%%%%%%%%%%%%%%%%%%%%%%%%%%%%%

\section*{Acknowledgment}

The authors would like to thank an anonymous reader for the remark on the usage of Carlini's lemma and of Kayal's 
implicit approach to give the alternative solution to the SET problem in the randomized case.

The third author would like to thank Michael Forbes, Ankit Gupta, Elad Haramaty, Swastik Kopparty, 
Ramprasad Saptharishi and Shubhangi Saraf for helpful discussions on related problems.

%%%%%%%%%%%%%%%%%%%%%%%%%%

\bibliographystyle{alpha}

\bibliography{shifts-eq-pit}
%%%%%%%%%%%%%%%%%%%%%%%%%%%%%%%%%%%%%%%%%%%%%%%%%%%

\appendix

%%%%%%%%%%%%%%%%%%%%%%%%%%%%%%%%%%%%%%%%%%%%%%%%%%%%
%%%%%%%%%%%%%%%%%%%%%%%%%%%%%%%%%%%%%%%%%%%%%%%%%%%%
%%%%%%%%%%%%%%%%%%%%%%%%%%
\section{Alternative Randomized Algorithm for SET}\label{sec:altmain}

In this section we give the alternative algorithm using Carlini's lemma and Kayal's approach in 
subsection~\ref{subsec:altproof}. In addition, we state Carlini's theorem
(as in \cite[Lemma 17]{kayal12}) for completeness.

\begin{lem}\label{lem:carlini}
	Given a polynomial $f(\bx) \in \F[\bx]$ with $m$ essential variables, we can compute in 
	randomized polynomial time an invertible linear transformation $A : \F^{(n \times n)^*}$ 
	such that $f(A\bx)$ depends on the first $m$ variables only.
\end{lem}

%%%%%%%%%%%%%%%%%%%%%%%%%%%%%%%%%%%%%%%%%%%%%%%%
\subsection{Formal Description}\label{subsec:altproof}

For simplicity, we will describe the algorithm receiving the input already preprocessed, where preprocessing is
done in the same way as in algorithm~\ref{alg:main}. Since the proof of correctness is analogous to the one in
section~\ref{sec:main}, we will not write the proof here.

\begin{algorithm}[H]
	\caption{Alternative Algorithm} \label{alg:alt}
 	\SetAlgoLined
	\KwIn{black-boxes (or white-boxes) for polynomials $f \in \cM_1 , g \in \cM_2$, and degree
	of $f$, which we denote by $d$.}
	\KwOut{a non-zero shift in $\cS_{f,g}$, if one exists, or FAIL, if $\cS_{f,g} = \emptyset$.}
	By interpolation, obtain black-box access to the homogeneous components $H^{d-1}(f), H^d(f)$ and
	$H^{d-1}(g), H^d(g)$. \\
	Find, via lemma~\ref{lem:carlini}, an invertible $n \times n$ matrix $A$ such that
	$H^d(f(A\bx))$ depends only on its essential variables (w.l.o.g., suppose that they are $x_1, \ldots, x_m$).
	Then, using Lemma~\ref{lem:poleq-rand}, solve the following system of equations, 
	where in these equations the variables are the entries of $\bb$:
	\begin{align*}
		H^{d}(g(A\bx)) &= H^{d}(f(A\bx)) \\
		H^{d-1}(g(A\bx)) &= H^{d-1}(f(A\bx)) + \sum_{k=1}^{n} b_k \cdot \frac{\partial H^d(f(A\bx))}{\partial x_k} \\
		b_k &= 0, \ \forall k > m. 
	\end{align*}
	If the system of equations above has no solution, \Return FAIL \\
	Else, proceed as follows: \\
	$f_1(\bx) \gets f(A\bx + \bb) - H^d(f(A\bx + \bb))$, \\ 
	$g_1(\bx) \gets g(A\bx) - H^d(f(A\bx + \bb))$, \\
	Recurse on this algorithm with input polynomials $f_1(\bx)$ and $g_1(\bx)$, and degree $d-1$. \\
	If the recursion returns FAIL, then \Return FAIL. \\
	Else, if recursion returns a shift $\bc$ such that $(c_1, \ldots, c_m) \neq (0, \ldots, 0)$, \Return FAIL. \\
	Else, take the shift $\bc$ and set $b_k = c_k$ for all $k > m$. \\
	Perform randomized PIT on $f(A\bx+\bb) - g(A\bx)$. \\
	If $f(A\bx+\bb) - g(A\bx) \equiv 0$, \Return $A^{-1}\bb$ \\
	Else, \Return FAIL
\end{algorithm}

\end{document}